\documentclass[12pt, draftclsnofoot, onecolumn]{IEEEtran}
\usepackage{epsfig,latexsym}
\usepackage{float}
\usepackage{indentfirst}
\usepackage{amsmath}
\usepackage{amssymb}
\usepackage{times}
\usepackage{subfigure}
\usepackage{psfrag}
\usepackage{cite}
\usepackage{lastpage}
\usepackage{fancyhdr}
\usepackage{color}
 \usepackage{amsthm}
\usepackage{bigints}
\newtheorem{theorem}{Theorem}
\newtheorem{Proposition}{Proposition}
\newtheorem{Lemma}{Lemma}
\newtheorem{lemma}[Lemma]{$\mathbf{Lemma}$}
\newtheorem{proposition}[Proposition]{Proposition}

\begin{document}%%
\title{ {    On the     Spectral Efficiency and Security Enhancements of NOMA Assisted Multicast-Unicast Streaming }}

\author{ Zhiguo Ding, \IEEEmembership{Senior Member, IEEE}, Zhongyuan Zhao, \IEEEmembership{Member, IEEE}, Mugen Peng, \IEEEmembership{Senior Member, IEEE},
  and  H. Vincent Poor, \IEEEmembership{Fellow, IEEE}\thanks{
Z. Ding and H. V. Poor  are with the Department of
Electrical Engineering, Princeton University, Princeton, NJ 08544,
USA.   Z. Ding is also with the School of
Computing and Communications, Lancaster
University, LA1 4WA, UK. Z. Zhao and M. Peng are  with the Key Laboratory of Universal Wireless Communications (Ministry
of Education), Beijing University of Posts and Telecommunications, Beijing,
China.
 }\vspace{-0.5em}} \maketitle
\begin{abstract}
 This paper considers the application of  non-orthogonal multiple access (NOMA) to a multi-user network with mixed multicasting and unicasting traffic. The proposed design of beamforming and power allocation  ensures that the unicasting performance is improved while maintaining the reception reliability of multicasting. Both analytical and simulation results are provided to demonstrate that the use of the NOMA assisted multicast-unicast scheme yields a significant improvement in spectral efficiency compared to orthogonal multiple access (OMA) schemes which realize  multicasting and unicasting services  separately. Since   unicasting messages are broadcasted to all the users, how   the use of NOMA can prevent those multicasting receivers intercepting the unicasting messages is also investigated, where it is shown that the secrecy unicasting rate achieved by NOMA is always larger than or equal to that of OMA. This security gain is mainly due to the fact that the  multicasting messages can be used as  jamming signals to prevent  potential eavesdropping  when the multicasting and unicasting messages are superimposed together following the NOMA principle.
\end{abstract}\vspace{-1em}
\section{Introduction}
Non-orthogonal multiple access (NOMA) has been recognized as an important enabling technology to realize   the challenging requirements of the fifth generation  (5G) mobile networks, such as massive connectivity, high data speed and low latency. The key idea of NOMA is to exploit the power domain for multiple access  and serve multiple users at the same time/frequency/code \cite{6692652, Nomading,7263349}. The two-user downlink special case of NOMA has been included in the 3rd Generation Partnership Project (3GPP) Long Term Evolution Advanced \cite{3gpp1}. In addition to its applications in cellular networks, NOAM has also been applied to other types of wireless networks, because of its superior spectral efficiency. For example, a variation of NOMA, termed Layer Division Multiplexing (LDM), has been proposed to the next general digital TV standard ATSC 3.0 \cite{7378924}.

Conventionally NOMA has been applied to unicasting transmission, where an information bearing message sent by the base station is  intended to  one receiver only. To these unicasting scenarios, various NOMA designs combined with multiple-input multiple-output (MIMO), millimeter-wave communications, and cooperative relaying have been developed \cite{7095538,Zhiguo_mmwave,7527682}.  Recently the application of NOMA to multicasting transmission has also attracted some attention, where one information bearing message is intended to  multiple users. For example, in \cite{7015589}, two types of messages are sent by the base station, where the high priority type of data is to be decoded by both the users, and the low priority type is   intended to one receiver only.

This paper is to consider the application of NOMA to a multi-user network with mixed multicasting and unicasting traffic, where the base station transmits two types of data streams, one for multicasting and one for unicasting. The study of this mixed multicast-unicast streaming is motivated by an important observation that in a multicasting  network, spatial degrees of freedom cannot be fully used. For example, a base station uses a beamformer to broadcast a multicasting message. In     rich scattering indoor environments, i.e., users' channels are independent from each other, it is inevitable that this beamformer is good to some users, but not so to  the others. Motivated by this inefficiency, in this paper, unicasting transmission is superimposed with multicasting following the NOMA principle, where the excess spatial degrees of freedom can be used to improve the performance of unicasting, while maintaining the reliability of multicasting.

In particular, the contribution of this paper is two-fold:
\begin{itemize}
\item The spectral efficiency of the proposed NOMA assisted multicast-unicasting scheme is characterized. Particularly, it is first shown that the proposed transmission scheme achieves the same multicasting performance as orthogonal multiple access (OMA) schemes which realize multicasting and unicasting services  separately. Then the reception reliability of NOMA unicasting is studied by using the outage probability as the criterion, and the unicasting  performance gain of NOMA over OMA is also investigated. The developed analytical and simulation results show that the use of NOMA can bring significant performance gains over OMA, and also provide specific guidelines for the design of user scheduling for further performance improvements.

\item Since the unicasting message is broadcasted to all the users, how well the use of NOMA can prevent those multicasting receivers intercepting the unicasting message is also investigated. First it is shown that the secrecy unicasting rate achieved by NOMA is always larger than or equal to that of OMA, and then the secrecy outage probability of NOMA unicasting is studied. Again the developed analytical results provide insights about how to design user scheduling in order to further enlarge the performance gap between NOMA and OMA. It is worth pointing out  that the reason for NOMA assisted multicast-unicast streaming  to achieve  better secrecy performance than OMA is similar to the idea of interference masking in conventional physical layer security networks \cite{6827055}. Particularly, the multicasting message can be viewed as a jamming signal, and this jamming signal effectively prevents those potential eavesdroppers with weak channel conditions to intercept the multicasting message.

\end{itemize}

\section{System Model}
Consider a downlink communication scenario with one base station communicating with $K$ users. The base station is equipped with $M$ antennas and each user has a single antenna. In this paper, we focus on the combination of multicasting and unicasting streaming, i.e., the base station has two messages to send. The multicasting message is intended to all the users, whereas the unicasting message is to be received by a particular user.

    In particular, denote the multicasting message sent from the base station by ${s}_M$. Without loss of generality, assume that the unicasting message, denoted by $ {s}_U$, is intended to user $1$. Using conventional OMA, two orthogonal resource blocks, such as time slots or frequency channels, are needed to deliver   the multicasting and unicasting messages separately. The use of the NOMA principle ensures that the multicast and unicast streaming services can be delivered within a single resource block.%  Since both messages are broadcasted simultaneously, it is important to ensure that user $k$, $k\neq 1$, will not intercept the message sent to user $1$, a physical layer security issue not been considered in \cite{Zhiguo_iot}.

 Particularly, with the application of the NOMA principle, the base station will transmit the following vector:
\begin{align}
\mathbf{x} = \mathbf{w}\left(\alpha_M s_{M}+ \alpha_U s_{U}  \right),
\end{align}
where $\mathbf{w}$ is an $M\times 1$ beamforming vector,  $\alpha_M$ and $\alpha_U$ are the  power allocation coefficients which  are designed to satisfy $\alpha_M^2+\alpha_{U}^2=1$.

 Following the MIMO-NOMA concept proposed in \cite{Zhiguo_iot}, we design the beamforming vector to artificially create the difference between the users' effective channel gains. Particularly, $\mathbf{w}$ is designed to improve the effective channel gain of user $1$, i.e.,
\begin{align}\label{bemforming}
\mathbf{w}=\frac{\mathbf{h}_1^H}{\sqrt{\mathbf{h}_1\mathbf{h}_1^H}},
\end{align}
where $\mathbf{h}_k$ denotes the $1\times M$ channel vector of user $k$. As a result, the NOMA principle can be applied even if the users have similar channel conditions. Note that the results developed in this paper about the spectral efficiency enhancement are new compared to those in \cite{Zhiguo_iot} due to the multi-user setup. In addition,  the security issue was also  not considered in the  existing work about MIMO-NOMA.

By using the beamforming design shown in \eqref{bemforming}, user $1$'s observation is given by
\begin{align}
{y}_1  &=\mathbf{h}_1\mathbf{x} +n_1  = \sqrt{\mathbf{h}_1\mathbf{h}_1^H} \left(\alpha_M s_{M}+ \alpha_U s_{U}  \right)+n_1,
\end{align}
where $ {n}_1$ is the additive Gaussian noise. Similar to a ``strong user" in conventional NOMA networks, user $1$ will carry out successive interference cancelation (SIC), i.e., $s_M$  is detected first and then subtracted from the observation before $s_U$ is decoded. Therefore, the signal-to-interference-plus-noise ratio (SINR) for user $1$ to detect $s_M$  is given by
\begin{align}\label{SINR}
{\text{SINR}}_{1} = \frac{\alpha_M^2 z_1}{\alpha_U^2z_1+\frac{1}{\rho} },
\end{align}
where $z_1=|\mathbf{h}_1|^2$ and $\rho$ is the transmit signal-to-noise ratio (SNR).   After $s_M$ is detected successfully, user $1$ first removes this message from its observation and then detects the unicasting message, $s_U$, with the following SNR:
\begin{align}\label{SNR}
\text{SNR}_{1} =\rho \alpha_U^2 z_1.
\end{align}

 User $k$'s observation, $2\leq k \leq K$,     is given by
\begin{align}
{y}_k  &=\mathbf{h}_k\mathbf{x} +n_k  = \frac{\mathbf{h}_k\mathbf{h}_1^H}{\sqrt{\mathbf{h}_1\mathbf{h}_1^H}} \left(\alpha_M s_{M}+ \alpha_U s_{U}  \right)+n_k.
\end{align}
Similar to a ``weak user" in  conventional NOMA networks, user $k$ detects $s_M$ by treating $s_U$ as noise, which means the SINR for detecting $s_M$ at user $k$ is given by
\begin{align}
\text{SINR}_{k} = \frac{ \alpha_M^2 z_{k} }{ \alpha_U^2z_{k}+\frac{1}{\rho}},
\end{align}
where $z_{k}=\frac{|\mathbf{h}_k\mathbf{h}_1^H|^2}{{\mathbf{h}_1\mathbf{h}_1^H}} $. Since $\mathbf{h}_1$ is independent from $\mathbf{h}_k$ and  a uniform transformation of a complex Gaussian vector is still complex Gaussian distributed, the probability density functions (pdfs) of $z_k$ are given by
 \begin{align}\label{pdf1}
f_{z_{k}}(z) =e^{-z},
 \end{align}
 for $2\leq k\leq K$,
 and
 \begin{align}\label{pdf2}
 f_{z_1}(z) = \frac{z^{M-1}}{(M-1)!}e^{-z},
 \end{align}
 respectively.

\subsection{Power allocation  to guarantee multicasting}
The proposed beamforming vector is helpful to increase the difference between the users' effective channel gains, which is ideal for the application of NOMA. However, it is important to design the power allocation policy in order to ensure that multicasting is delivered successfully. In this paper,   the cognitive radio inspired power allocation policy \cite{Zhiguo_CRconoma} is used by treating $s_M$ as the message to be broadcasted to the primary users, which means that $s_M$ is assigned  with  a higher priority compared to $s_U$. Particularly, to ensure all the users to receive the multicasting message $s_M$ correctly, we impose the following constraint on the power allocation coefficients:
\begin{align}\label{constraint NOMA}
\log(1+{\text{SINR}}_{k})\geq R_{M}
\end{align}
for all $k\in \{1, \cdots, K\}$, where $R_M$ is the targeted data rate for multicasting. Therefore the power allocation coefficient can be set as follows:
\begin{align}
\alpha_U^2 = &\max \left\{0, \min  \left\{\frac{z_{k}-\frac{\epsilon_{M}}{\rho}}{z_{k}(1+\epsilon_{M})}, 1\leq k \leq K \right\} \right\},
\end{align}
where $\epsilon_M=2^{R_M}-1$.
It is worth pointing out that $z_1$ and $z_k$, $k>1$, are distributed differently as shown in \eqref{pdf1} and \eqref{pdf2}.

As a result, the unicasting data rate  at user $1$ achieved by the NOMA scheme  is given by
\begin{align}\label{xxxx1}
R_{U,1} =&\log\left(1+\rho z_1 \max \left\{0, \right.\right. \\ \nonumber &\left.\left.\min  \left\{\frac{z_{k}-\frac{\epsilon_{M}}{\rho}}{z_{k}(1+\epsilon_{M})}, 1\leq k \leq K \right\} \right\}\right).
\end{align}
It is important to point out that, in \eqref{xxxx1},  the possible failure of the first SIC step has already  been taken into the consideration. For example, as shown in \eqref{xxxx1}, $R_{U,1}$ can be zero. This case will happen if one of the users in the network experiences deep fading. As a result, the base station allocates all the power for multicasting, and the rate for unicasting will be zero.

Similarly, the eavesdropping rate for user $k$, $k>1$, to intercept $S_U$ is given by
\begin{align}
R_{U,k} =&\log\left(1+\rho z_k \max \left\{0, \right.\right. \\ \nonumber &\left.\left.\min  \left\{\frac{z_{k}-\frac{\epsilon_{M}}{\rho}}{z_{k}(1+\epsilon_{M})}, 1\leq k \leq K \right\} \right\}\right),
\end{align}
for $2\leq k \leq K$. Ideally the difference between $R_{U,1}$ and $R_{U,k}$, $k>1$, should be kept as large as possible, which makes user $k$, $k>1$, difficult to decode $s_U$. This security  issue will be studied in Section \ref{section security}.

\subsection{A sophisticated  OMA-based benchmarking scheme  }\label{subsection beamforming}
There are two types of OMA  transmission schemes which can be used as benchmarking schemes.
 One is based on the use of predefined orthogonal bandwidth blocks, such as time slots with fixed durations or frequency channels with fixed bandwidth. The other is to dynamically adjust  the amount of bandwidth resources allocated for multicasting and unicasting   according to the users' channel conditions. In this paper, we use the latter as a benchmark since it outperforms  the former. But it is important to point out that this sophisticated OMA scheme is difficult to   implement since high-cost circuits are needed to support the OMA scheme using time slots (frequency channels) with arbitrary durations (bandwidth).

Without loss of generality, time division multiple access (TDMA) is used as a representative of OMA. Similar to the cognitive radio inspired NOMA power allocation policy, in OMA, a portion of the whole time slot, denoted by $\gamma$, $0< \gamma\leq1$, is allocated to transmit the multicasting message, $s_M$.  If $\gamma\neq 1$, the remaining time will be used to transmit the unicasting message, $s_U$. During the multicasting phase, the base station uses $\mathbf{p}=\frac{\mathbf{h}_1^H}{\sqrt{\mathbf{h}_1\mathbf{h}_1^H}}$ as the beamforming vector for multicasting. Note that the base station can  use other choices, such as equal gain combining based beamforming, i.e., $\mathbf{p}= \frac{1}{\sqrt{M}}\begin{bmatrix}1 &\cdots & 1  \end{bmatrix}^T$ or a randomly chosen vector. The simulation results provided in Section \ref{section simulation} demonstrate that different choices of beamforming result in similar  performance. It is worth pointing out that   a beamforming choice of $\mathbf{p}=\frac{\mathbf{h}_1^H}{\sqrt{\mathbf{h}_1\mathbf{h}_1^H}}$ slightly outperforms the other two, which means that this is a choice preferred by OMA.

      Similar to \eqref{constraint NOMA}, the requirement that all the users can receive $s_M$ results in the following constraint on the time allocation coefficient:
\begin{align}\label{constraint NOMA}
\gamma\log(1+\rho \min\{|\mathbf{h}_k\mathbf{p}|^2, 1\leq k\leq K\})\geq R_{M}.
\end{align}
     Therefore the time allocation coefficient can be set as follows:
\begin{align}
\gamma=\min\left\{1, \frac{R_{M}}{\log(1+\rho \min\{|\mathbf{h}_k\mathbf{p}|^2, 1\leq k\leq K\})}\right\}.
\end{align}
The remaining $(1-\gamma)$ duration is used for unicasting by again employing  the precoding vector $\mathbf{p}=\frac{\mathbf{h}_1^H}{\sqrt{\mathbf{h}_1\mathbf{h}_1^H}}$, which means that the following data rate is achievable for unicasting at user $k$:
\begin{align}
\bar{R}_{U,k} =& \log\left(1+\rho z_k \right)\left(1-\right. \\ \nonumber &\left.\min\left\{1, \frac{R_{M}}{\log(1+\rho \min\{|\mathbf{h}_k\mathbf{p}|^2, 1\leq k\leq K\})}\right\}\right).
\end{align}
Again note that $\bar{R}_{U,k}$ can be zero, if there is a user experiencing deep fading and all the time is used for multicasting.

\section{Spectral Efficiency Enhancements Achieved by NOMA}\label{section spectral}
Note that only the unicasting performance is focused in this paper, since both the NOMA and OMA schemes achieve the same multicasting performance, as illustrated in the following.
\begin{proposition}\label{propostion 1}
The outage probability for NOMA multicasting is the same as that for OMA multicasting.
\end{proposition}
\begin{proof}
The outage probability for NOMA multicasting is given by
\begin{align}
\mathrm{P}^o_M &\triangleq \mathrm{P}\left( \log(1+\text{SINR}_k)<R_M, 1\leq k \leq K  \right) \\ \nonumber & = \mathrm{P}\left(\alpha_U^2=0\right) =\mathrm{P}\left(\min\left\{z_k, 1\leq k \leq K \right\}< \frac{\epsilon_M}{\rho}\right) .
\end{align}
Similarly, for the OMA scheme, its multicasting outage probability is expressed as follows:
\begin{align}
\mathrm{P}^n_M  &= \mathrm{P}\left(\gamma=1\right)\\ \nonumber & =\mathrm{P}\left( \frac{R_{M}}{\log(1+\rho \min\{z_k, 1\leq k\leq K\})}>1\right) .
\end{align}
With some algebraic manipulations, it is straightforward to show that $\mathrm{P}^o_M=\mathrm{P}^n_M$, and the proof is complete.
\end{proof}
Therefore, in the remaining of this paper, we will focus on the unicasting performance. Particularly, in this section, two criteria will be used to study the spectral efficiency of the proposed NOMA based unicasting scheme. One is the outage probability achieved by the proposed scheme, i.e., $\mathrm{P}(R_{U,1}< R_U)$, where $R_U$ denotes the targeted data rate for unicasting. The other is the comparison between two instantaneous unicasting rates achieved by NOMA and OMA, i.e., $\mathrm{P}(R_{U,1}>\bar{R}_{U,1})$, .

\subsection{Characterizing the Unicasting Outage Probability}
Recall that the unicasting outage probability for the NOMA scheme can be written as follows:
\begin{align}
\mathrm{P}_N = &\mathrm{P}\left( z_1 \max \left\{0, \right.\right. \\ \nonumber &\left.\left.\min  \left\{\frac{z_{k}-\frac{\epsilon_{M}}{\rho}}{z_{k}(1+\epsilon_{M})}, 1\leq k \leq K \right\} \right\}  < \frac{\epsilon_{U}}{\rho}\right),
\end{align}
where $\epsilon_U=2^{R_U}-1$. The following theorem provides a closed-form expression for this outage probability.

\begin{theorem}
The unicasting outage probability  achieved by the proposed NOMA transmission scheme can be approximated as follows:
\begin{align} \nonumber
\mathrm{P}_N &\approx 1 - \frac{\Gamma\left(M, \frac{\epsilon_M}{\rho}\right)}{(M-1)!} e^{-\frac{(K-1)\epsilon_M}{\rho}}  +\frac{\gamma(M, K\phi) - \gamma(M, \frac{K\epsilon_M}{\rho})}{(M-1)!K^M} \\ \nonumber &
+\sum^{N_a}_{i=1}w_i\frac{b-a}{2}\left[
 F_{\tilde{z}_1}\left( \frac{b-a}{2}x_i+\frac{a+b}{2} \right)  \right.\\ \nonumber &\left.- F_{\tilde{z}_1}\left(\frac{1}{\psi}-\frac{\epsilon_{M}}{\rho\psi}\left(\frac{b-a}{2}x_i+\frac{a+b}{2}\right) \right) \right]\\   &\times  f_{\tilde{u}}\left(\frac{b-a}{2}x_i+\frac{a+b}{2}\right) \sqrt{1-x_i^2}
,
\end{align}
where $\phi=\frac{\epsilon_{M}}{\rho}+\frac{\epsilon_{U}(1+\epsilon_{M})}{\rho}$, $\psi=\frac{\epsilon_{U}(1+\epsilon_{M})}{\rho}$, $a=\frac{1}{\psi\left(1+\frac{\epsilon_{M}}{\rho\psi}\right)}$, $b=\frac{\rho}{\epsilon_M}$, $x_i =\cos\left(\frac{2i-1}{2N_a}\pi\right)$, $w_i=\frac{\pi}{N_a}$, $
F_{\tilde{z}_1}(z) = \frac{\Gamma(M, \frac{1}{z})}{(M-1)!}$,  $f_{\tilde{u}}(x) = \frac{1-K}{x^2}e^{-\frac{K-1}{x}}$, $N_a$ denotes the parameter of the Chebyshev-Gauss approximation,  $\Gamma(\cdot)$ and $\gamma(\cdot)$ denote the upper and lower incomplete gamma functions, respectively.
\end{theorem}
\begin{proof}
Recall that $z_k$, $2\leq k \leq K$, are independent and identically distributed, and   $z_1$ is independent from $z_k$, $2\leq k \leq K$. Another important   fact is that    $f(y) \triangleq  \frac{y-\frac{\epsilon_{M}}{\rho}}{y(1+\epsilon_{M})}$ is a monotonically increasing function of $y$, for $y>0$, which can be verified as follows:
\begin{align}
f'(y) = \frac{\epsilon_M}{y^2(1+\epsilon_M)\rho}>0,
\end{align}
for $y>0$.  Therefore, define $u=\min\{z_2, \cdots, z_K\}$ whose pdf is given by \cite{David03}
\begin{align}
f_{u} (x) = (K-1)e^{-(K-1)x},
\end{align}
and the outage probability can be expressed as follows:
\begin{align}
\mathrm{P}_N = &\mathrm{P}\left( z_1 \max \left\{0, \right.\right. \\ \nonumber &\left.\left.\min  \left\{\frac{z_{1}-\frac{\epsilon_{M}}{\rho}}{z_{1}(1+\epsilon_{M})}, \frac{u-\frac{\epsilon_{M}}{\rho}}{u(1+\epsilon_{M})} \right\} \right\}  < \frac{\epsilon_{U}}{\rho}\right),
\end{align}
which can be further separated into three terms as follows:
\begin{align}\nonumber
\mathrm{P}_N = &\mathrm{P}\left(\min \{z_1, u\}<\frac{\epsilon_M}{\rho}\right) +\mathrm{P}\left( z_1  \right. \\ \nonumber &\left. \times \min  \left\{\frac{z_{1}-\frac{\epsilon_{M}}{\rho}}{z_{1}(1+\epsilon_{M})}, \frac{u-\frac{\epsilon_{M}}{\rho}}{u(1+\epsilon_{M})}   \right\}  < \frac{\epsilon_{U}}{\rho}\right)\\ \label{threem terms}  =&\underset{Q_1}{\underbrace{\mathrm{P}\left(\min \{z_1, u\}<\frac{\epsilon_M}{\rho}\right)}} \\ \nonumber &+\underset{Q_2}{\underbrace{\mathrm{P}\left(z_1>\frac{\epsilon_M}{\rho},z_1<u,      \frac{z_{1}-\frac{\epsilon_{M}}{\rho}}{(1+\epsilon_{M})}    < \frac{\epsilon_{U}}{\rho}\right)}} \\ \nonumber &+\underset{Q_3}{\underbrace{\mathrm{P}\left(u>\frac{\epsilon_M}{\rho}, z_1>u, z_1   \frac{u-\frac{\epsilon_{M}}{\rho}}{u(1+\epsilon_{M})}      < \frac{\epsilon_{U}}{\rho}\right)}} .
\end{align}

The first term in the above expression can be  found as follows:
\begin{align}\label{equ1}
Q_1 &= 1- \mathrm{P}\left(\min \{z_1, u\}>\frac{\epsilon_M}{\rho}\right)\\ \nonumber & =1 - \frac{\Gamma\left(M, \frac{\epsilon_M}{\rho}\right)}{(M-1)!} e^{-\frac{(K-1)\epsilon_M}{\rho}}   .
\end{align}.

The second term in \eqref{threem terms} can be calculated as follows:
\begin{align}
Q_2 =&\mathrm{P}\left(z_1>\frac{\epsilon_M}{\rho}, z_1<u,    z_{1}     < \phi\right)\\ \nonumber =&\mathrm{P}\left(z_1>\frac{\epsilon_M}{\rho}, z_1<u < \phi\right)  +\mathrm{P}\left(z_1>\frac{\epsilon_M}{\rho}, z_1    < \phi<u\right).
\end{align}

By using the pdfs of $u$ and $z_1$ in \eqref{pdf1} and \eqref{pdf2},  $Q_2$ can be found as follows:
\begin{align}\nonumber
Q_2  =&\int^{\phi}_{\frac{\epsilon_M}{\rho}}\int^{\phi}_{x} f_u(y)dy f_{z_1}(x)dx  \\ \nonumber &+\mathrm{P}\left(\frac{\epsilon_M}{\rho}<z_1    < \phi\right)\mathrm{P}\left( u>\phi\right)\\ \nonumber =& \int^{\phi}_{\frac{\epsilon_M}{\rho}} \left( e^{-(K-1)x}-e^{-(K-1)\phi}   \right) f_{z_1}(x)dx \\ \nonumber &+ \frac{\gamma(M, \phi)-\gamma(M,\frac{\epsilon_M}{\rho})}{(M-1)!}  e^{-(K-1)\phi}.
\end{align}
With some algebraic manipulations, we can find $Q_2$ in the following closed-form expression:
\begin{align}\label{equ2}
Q_2   %=&   (1-e^{-\phi})^{K-1}\frac{\gamma(M, \phi)-\gamma(M,\frac{\epsilon_M}{\rho})}{(M-1)!} - \sum^{K-1}_{m=0}{K-1 \choose m} (-1)^m \\ \nonumber &\times \int^{\phi}_{\frac{\epsilon_M}{\rho}} e^{-mx} f_{z_1}(x)dx  + \frac{\gamma(M, \phi)-\gamma(M,\frac{\epsilon_M}{\rho})}{(M-1)!} \left[1-(1-e^{-\phi})^{K-1}\right]\\ \nonumber
=&\frac{\gamma(M, K\phi) - \gamma(M, \frac{K\epsilon_M}{\rho})}{(M-1)!K^M}.
\end{align}

The last term in \eqref{threem terms} can be expressed as follows:
\begin{align}\nonumber
Q_3=&\mathrm{P}\left( z_1>u, z_1   \frac{u-\frac{\epsilon_{M}}{\rho}}{u(1+\epsilon_{M})}      < \frac{\epsilon_{U}}{\rho}, u>\frac{\epsilon_M}{\rho}\right) \\ \nonumber=&
\mathrm{P}\left( \frac{1}{z_1}<\frac{1}{u},    1-\frac{\epsilon_{M}}{u\rho}      < \frac{\epsilon_{U}(1+\epsilon_{M})}{z_1\rho}, u>\frac{\epsilon_M}{\rho}\right).
\end{align}
 Furthermore, define $\tilde{z}_1=\frac{1}{z}$ and $\tilde{u}=\frac{1}{u}$. Therefore the CDF of $\tilde{z}_1$ and the pdf of $\tilde{u}$ are  $F_{\tilde{z}_1}(z)$ and
$
f_{\tilde{u}}(x)$ defined in the theorem, respectively. Therefore, the factor, $Q_3$, can be rewritten as follows:
\begin{align}\nonumber
Q_3 =&
\mathrm{P}\left( \tilde{z}_1<\tilde{u},    1-\frac{\epsilon_{M}}{\rho}\tilde{u}      < \psi\tilde{z}_1,\tilde{z}_1<\frac{\rho}{\epsilon_M},\tilde{u}<\frac{\rho}{\epsilon_M}\right)\\ \nonumber =&
\mathrm{P}\left( \tilde{z}_1<\tilde{u},   \tilde{z}_1> \frac{1}{\psi}-\frac{\epsilon_{M}}{\rho\psi}\tilde{u}       ,\tilde{z}_1<\frac{\rho}{\epsilon_M},\tilde{u}<\frac{\rho}{\epsilon_M}\right)
\\ \nonumber =&
\mathrm{P}\left(\frac{1}{\psi}-\frac{\epsilon_{M}}{\rho\psi}\tilde{u}  <  \tilde{z}_1< \tilde{u} ,\tilde{u}<\frac{\rho}{\epsilon_M}\right).
\end{align}
Note that the constraint of $\frac{1}{\psi}-\frac{\epsilon_{M}}{\rho\psi}\tilde{u} < \tilde{u}$ results in the following additional constraint on $\tilde{u}$
\begin{align}
\tilde{u}>\frac{1}{\psi\left(1+\frac{\epsilon_{M}}{\rho\psi}\right)}.
\end{align}
Therefore, $Q_3$ can be calculated as follows:
\begin{align}
Q_3 =& \int_{\frac{1}{\psi\left(1+\frac{\epsilon_{M}}{\rho\psi}\right)}}^{\frac{\rho}{\epsilon_M}}\left(
 F_{\tilde{z}_1}\left( x \right) - F_{\tilde{z}_1}\left(\frac{1}{\psi}-\frac{\epsilon_{M}}{\rho\psi}x \right) \right) f_{\tilde{u}}(x)dx.
\end{align}
Finding an exact expression for the above integral is difficult. In order to apply  Chebyshev-Gauss quadrature, the above integral can be first rewritten as follows:
\begin{align}
Q_3 =& \frac{b-a}{2}\int^{1}_{-1}\left[
 F_{\tilde{z}_1}\left( \frac{b-a}{2}x+\frac{a+b}{2} \right) - F_{\tilde{z}_1}\left(\frac{1}{\psi}\right.\right.\\ \nonumber &\left.\left.-\frac{\epsilon_{M}}{\rho\psi}\left(\frac{b-a}{2}x+\frac{a+b}{2}\right) \right) \right] f_{\tilde{u}}\left(\frac{b-a}{2}x+\frac{a+b}{2}\right)dx.
\end{align}

After  applying  Chebyshev-Gauss quadrature,  $Q_3$ can be approximated as follows:
\begin{align}\label{equ3}
Q_3 \approx & \sum^{N_a}_{i=1}w_i\frac{b-a}{2}\left[
 F_{\tilde{z}_1}\left( \frac{b-a}{2}x_i+\frac{a+b}{2} \right)  \right.\\ \nonumber &\left.- F_{\tilde{z}_1}\left(\frac{1}{\psi}-\frac{\epsilon_{M}}{\rho\psi}\left(\frac{b-a}{2}x_i+\frac{a+b}{2}\right) \right) \right]\\ \nonumber &\times  f_{\tilde{u}}\left(\frac{b-a}{2}x_i+\frac{a+b}{2}\right) \sqrt{1-x_i^2}.
\end{align}
Substituting \eqref{equ1}, \eqref{equ2} and \eqref{equ3} into \eqref{threem terms}, a closed-form expression for the outage probability can be obtained and  the theorem is proved.
\end{proof}
The steps used to obtain  the closed form expression provided in the above theorem can also be used to calculate the achievable diversity gain, as shown in the following lemma.
\begin{lemma}\label{lemma1}
The diversity gain  for unicasting transmission achieved by the proposed NOMA scheme is $1$.
\end{lemma}
\begin{proof}
The diversity gain achieved by the proposed NOMA scheme can be obtained by studying the upper and lower bounds on the outage probability. Based on the expression for the outage probability shown in \eqref{threem terms}, we can obtain the  following  lower bound on the outage probability:
\begin{align}\nonumber
\mathrm{P}_N &\geq Q_1\underset{(a)}{=} 1 -   \left(e^{-\frac{\epsilon_M}{\rho}} \sum^{M-1}_{m=0}\frac{\epsilon_M^m}{\rho^mm!}\right)  e^{-\frac{(K-1)\epsilon_M}{\rho}}  \\   &\underset{(b)}{\approx}  1 -   \left(1-\frac{K\epsilon_M}{\rho}\right) \label{Q1} = \frac{K\epsilon_M}{\rho},
\end{align}
where   step $(a)$ follows from Eq. (8.352.2) in \cite{GRADSHTEYN} and step $(b)$ follows from the high SNR approximation. The approximation in \eqref{Q1} implies  that the diversity gain achieved by the NOMA scheme is upper bounded by $1$.

On the other hand, we can construct the following upper bound on the outage probability:
\begin{align}\label{q124}
\mathrm{P}_N &\leq Q_1+Q_2+Q_4,
\end{align}
where
\begin{align}\nonumber
Q_{31}=&\mathrm{P}\left( z_1>u,   \frac{u-\frac{\epsilon_{M}}{\rho}}{(1+\epsilon_{M})}      < \frac{\epsilon_{U}}{\rho}, u>\frac{\epsilon_M}{\rho}\right)  .
\end{align}
The probability in \eqref{q124} is an upper bound on $\mathrm{P}_N$ since $Q_{31}>Q_3$. Note that $Q_{31}$ can be calculated as follows:
\begin{align} \label{Q4}
Q_{31}=&\mathrm{P}\left( z_1>u,    u      < \frac{\epsilon_{M}}{\rho}+\psi, u>\frac{\epsilon_M}{\rho}\right)\\ \nonumber
\leq&\mathrm{P}\left(     u      < \frac{\epsilon_{M}}{\rho}+\psi, u>\frac{\epsilon_M}{\rho}\right)
\\\nonumber =& e^{-(K-1)\frac{\epsilon_M}{\rho}} - e^{-(K-1)\left(\frac{\epsilon_{M}}{\rho}+\psi\right)} \approx
 (K-1)\psi,
\end{align}
where the approximation is obtained in the high SNR regime. Furthermore,  $Q_2$ can be upper bounded as follows:
 \begin{align} \label{Q2x}
Q_2
\leq&\frac{\gamma(M, K\phi)}{(M-1)!K^M}\\ \nonumber =& \frac{1-e^{-K\phi}\sum^{M-1}_{m=0}\frac{K^m\phi^M}{m!} }{K^M}\approx  \frac{K\phi }{K^M},
\end{align}
where the series representation of gamma functions based on Eq. (8.352.2) in \cite{GRADSHTEYN} has been used.
Combining \eqref{q124}, \eqref{Q1}, \eqref{Q2x} and \eqref{Q4}, one can find that the diversity order achieved by the NOMA scheme is lower bounded by one. Since both the upper and lower bounds on the diversity gain are one, the proof is complete.
\end{proof}
{\it Remark 1:} The reason to have a diversity gain of $1$ can be explained in the following. Because of the used cognitive radio power allocation policy, the bottleneck of the system is the quality of the weakest channel gain, $u$. If $u$ is smaller than $\frac{\epsilon_M}{\rho}$, i.e., the user with the weakest channel condition cannot detect the multicasting message correctly, all the power will be spent for multicasting, which is the dominant event among all the possible outage events for unicasting. Following steps similar to those in the proof for Lemma \ref{lemma1}, it is straightforward to show that the probability for the event of $u<\frac{\epsilon_M}{\rho}$ is inversely proportional to the SNR, i.e., a diversity gain of $1$.  It is worth noting that the result shown in Lemma \ref{lemma1} is consistent to the one  previously reported in \cite{Zhiguo_CRconoma}.

\subsection{Performance Gain of NOMA over OMA}\label{subsection spec 2}
In this subsection, the likelihood that NOMA will outperform OMA is studied first, which offers some insights about how to design  user scheduling in order to further enlarge the performance gap between NOMA and OMA.
The probability for  NOMA to outperform  OMA  can be characterized  as follows:
\begin{align}
\mathrm{P}_D \triangleq &\mathrm{P}\left(R_{U,1}-\bar{R}_{U,1}\leq0\right)
 \\ \nonumber
 =&\mathrm{P}\left(\log\left(1+\rho z_1 \max \left\{0, \right.\right.\right. \\ \nonumber &\left.\left.\min  \left\{\frac{z_{1}-\frac{\epsilon_{M}}{\rho}}{z_{1}(1+\epsilon_{M})},\frac{u-\frac{\epsilon_{M}}{\rho}}{u(1+\epsilon_{M})} \right\} \right\}\right)
\\\nonumber
 &\leq \log\left(1+\rho z_1 \right) \left(1-\min\left\{1,\right.\right.\\ \nonumber &  \left.\left.\left. \frac{R_{M}}{\log(1+\rho \min\{|\mathbf{h}_k\mathbf{p}|^2, 1\leq k\leq K\})}\right\}\right)\right).
\end{align}

  Note that in the case of all the power (time) is allocated to multicasting, the two schemes realize the same performance, which means that the addressed probability  can be rewritten as follows:
\begin{align}\nonumber
\mathrm{P}_D
 &=\mathrm{P}\left(\min\{z_1,u\}<\frac{\epsilon_M}{\rho}\right)+\mathrm{P}\left(\min\{z_1,u\}>\frac{\epsilon_M}{\rho},    \right. \\ \nonumber &\log\left(1+\rho z_1 \min  \left\{\frac{z_{1}-\frac{\epsilon_{M}}{\rho}}{z_{1}(1+\epsilon_{M})},\frac{u-\frac{\epsilon_{M}}{\rho}}{u(1+\epsilon_{M})}   \right\}\right)
\\
 &\leq \log\left(1+\rho z_1 \right)   \left.\left(1-  \frac{R_{M}}{\log(1+\rho \min\{z_1, u\})} \right)\right),
\end{align}
where $|\mathbf{h}_k\mathbf{p}|^2$ is equal to $z_k$ since $\mathbf{p}=\mathbf{w}$.
Depending on the relationship between $u$ and $z_1$, we can further separate the probability into the following terms:
\begin{align}\label{y1}
\mathrm{P}_D
 =&\mathrm{P}\left(\min\{z_1,u\}<\frac{\epsilon_M}{\rho}\right)\\ \nonumber &+\mathrm{P}\left(u>\frac{\epsilon_M}{\rho}, z_1>u,    \right.  \log\left(1+\rho z_1 \frac{u-\frac{\epsilon_{M}}{\rho}}{u(1+\epsilon_{M})}    \right)
\\\nonumber
 &\leq \log\left(1+\rho z_1 \right)   \left.\left(1-  \frac{R_{M}}{\log(1+\rho u)} \right)\right)\\ \nonumber
 &+\mathrm{P}\left(z_1>\frac{\epsilon_M}{\rho},  z_1<u,  \right.  \log\left(1+\rho  \frac{z_{1}-\frac{\epsilon_{M}}{\rho}}{ (1+\epsilon_{M})} \right)
\\\nonumber
 &\leq \log\left(1+\rho z_1 \right)   \left.\left(1-  \frac{R_{M}}{\log(1+\rho z_1)} \right)\right).
\end{align}
One  can evaluate that the following equality always holds:
\begin{align}
\label{y2}
  \log\left(1+\rho  \frac{z_{1}-\frac{\epsilon_{M}}{\rho}}{ (1+\epsilon_{M})} \right)
 =\log\left(1+\rho z_1 \right)    \left(1-  \frac{R_{M}}{\log(1+\rho z_1)} \right),
\end{align}
which means the probability, $\mathrm{P}_D$, is lower bounded by the following:
\begin{align}
\mathrm{P}_D
 \geq&\mathrm{P}\left(z_1>\frac{\epsilon_M}{\rho},  z_1<u \right)\\\nonumber =& \int^{\infty}_{\frac{\epsilon_M}{\rho}} e^{-Kz}\frac{z^{M-1}}{(M-1)!} dz = \frac{\Gamma\left(M, \frac{K\epsilon_M}{\rho}\right)}{(M-1)!K^M} .
\end{align}
This lower bound  can be approximated at high SNR as follows:
\begin{align}
\mathrm{P}_D
 \geq&e^{-K\frac{\epsilon_M}{\rho}}\sum^{M-1}_{m=0}\frac{\left(\frac{\epsilon_M}{\rho}\right)^{m}}{K^{M-m}m!}\approx
   \frac{1}{K^{M}},
\end{align}
which means that it is always possible that the unicasting rate of NOMA is smaller than that of OMA, even at high SNR.

{\it Remark 2:} An important conclusion  from the above analysis is that  the event of $z_1<u$ is very damaging to the performance of NOMA. Particularly,  \eqref{y1} and \eqref{y2} show that the event of $z_1<u$ leads to the situation that NOMA offers no performance gain over OMA. This observation motivates the following user scheduling scheme.

\underline{User Scheduling:} Prior to the NOMA transmission, the base station selects a user whose channel norm is the largest for unicasting, i.e., user $i^*$ is scheduled for unicasting if $i^*=\arg \max \{|\mathbf{h}_k|^2, 1\leq k \leq K\}$.

With such a choice of $i^*$, the case of $z_1<u$ can be avoided since
\begin{align}
u&\triangleq  |\mathbf{h}_u\mathbf{w}|^2 \underset{(a)}{\leq}   |\mathbf{w}|^2 |\mathbf{h}_u|^2 \underset{(b)}{=} |\mathbf{h}_u|^2 \underset{(c)}{\leq} |\mathbf{h}_{i^*}|^2\triangleq   {z_1},
\end{align}
where $\mathbf{h}_u$ denotes the channel vector for the user with the smallest channel norm, step (a) follows from the Cauchy-Schwarz inequality, step (b) follows from the fact that $\mathbf{w}=\frac{\mathbf{h}_{i^*}^H}{|\mathbf{h}_{i^*}|}$ and step (c) is due to the used scheduling scheme. The simulation results provided in Section \ref{section simulation} demonstrate that the use of this user scheduling scheme effectively increases the performance gap between NOMA and OMA.

 \section{Security  Enhancements Achieved by NOMA}\label{section security}
 In this section, we first show that the use of NOMA unicasting can always improve the unicasting security, compared to OMA, and then the unicasting secrecy outage probability is studied, from which insights about how to further improve the security enhancements of NOMA can be obtained.
 \subsection{The reduction of the eavesdropping capability by using NOMA}
 First define the secrecy rates achieved by NOMA and OMA as follows:
   \begin{align}
  R_S\triangleq &\left(R_{U,1} - \max\{R_{U,k}, 2\leq k\leq K\}\right)^+
   \end{align}
   and
   \begin{align} \bar{R}_S\triangleq&\left(\bar{R}_{U,1} - \max\{\bar{R}_{U,k}, 2\leq k\leq K\}\right)^+,
  \end{align}
 respectively,  where $(x)^+\triangleq \max\{0, x\}$.

In order to show $R_S$ is always larger than or equal to $\bar{R}_S$, i.e., $R_S\geq \bar{R}_S$,  the following lemma is presented first.
  \begin{lemma}\label{lemma2}
  Define the function $F(x)$, for $x\geq u$, as follows:
  \begin{align}\label{FU}
  F_u(x)=&\log\left(1+\rho x    \frac{u-\frac{\epsilon_{M}}{\rho}}{u(1+\epsilon_{M})}\right) \\\nonumber
  &-\left(1-   \frac{R_{M}}{\log(1+\rho u)} \right) \log\left(1+\rho x \right)
  \end{align}
  where $u>\frac{\epsilon_M}{\rho}$.
  This is  a monotonically decreasing function with respect to  $x$,   in the high SNR regime.
  \end{lemma}
  \begin{proof}
  To simplify the proof, we rewrite the function as follows:
    \begin{align}
  F_u(x)=&\log e \cdot \ln\left(1+\rho x    \frac{u-\frac{\epsilon_{M}}{\rho}}{u(1+\epsilon_{M})}\right) \\\nonumber
  &- \log e \cdot \ln\left(1+\rho x \right)^{\left(1-   \frac{R_{M}}{\log e \cdot \ln(1+\rho u)} \right)}.
  \end{align}
The lemma can be proved by showing that the first order derivative of the function  is negative. In particular, the first order derivative of $F_u(x)$ is given by
  \begin{align}
  \frac{d F_u(x)}{dx}=&\log e \cdot  \frac{\rho     \frac{u-\frac{\epsilon_{M}}{\rho}}{u(1+\epsilon_{M})}}{1+\rho x    \frac{u-\frac{\epsilon_{M}}{\rho}}{u(1+\epsilon_{M})}} - \rho \log e  \\\nonumber
  &\times \frac{\left(1-   \frac{R_{M}}{\log e \cdot \ln(1+\rho u)} \right)\left(1+\rho x \right)^{\left(-   \frac{R_{M}}{\log e \cdot \ln(1+\rho u)} \right)}}{\left(1+\rho x \right)^{\left(1-   \frac{R_{M}}{\log e \cdot \ln(1+\rho u)} \right)}}.
  \end{align}
In order to show $\frac{d F_u(x)}{dx}\leq 0$,  we first have the following:
  \begin{align}
  &\frac{\left(1+\rho x \right)\left(1+\rho x    \frac{u-\frac{\epsilon_{M}}{\rho}}{u(1+\epsilon_{M})}\right)}{\rho\log e}\frac{d F_u(x)}{dx}\\\nonumber =&       \left(1+\rho x \right)   \frac{u-\frac{\epsilon_{M}}{\rho}}{u(1+\epsilon_{M})} -       \left(1+\rho x    \frac{u-\frac{\epsilon_{M}}{\rho}}{u(1+\epsilon_{M})}\right)\\ \nonumber &\times \left(1-   \frac{R_{M}}{\log e \cdot \ln(1+\rho u)} \right)\\\nonumber =&  \frac{u-\frac{\epsilon_{M}}{\rho}}{u(1+\epsilon_{M})}+ \frac{R_{M}}{\log(1+\rho u)} \\ \nonumber &+\rho x \frac{u-\frac{\epsilon_{M}}{\rho}}{u(1+\epsilon_{M})}  \frac{R_{M}}{\log(1+\rho u)}  -1 .
  \end{align}
  With fixed $u$ and $x$, by increasing $\rho$, we can have the following approximation:
    \begin{align}
  &\frac{\left(1+\rho x \right)\left(1+\rho x    \frac{u-\frac{\epsilon_{M}}{\rho}}{u(1+\epsilon_{M})}\right)}{\rho\log e}\frac{d F_u(x)}{dx} \\\nonumber \rightarrow &  \frac{1}{(1+\epsilon_{M})} \left(1+   \frac{\rho x R_{M}}{\log(\rho u)}\right)  -1 .
  \end{align}
  Note that when $\rho\rightarrow \infty $, we can have $ \frac{\rho }{\log\rho }\rightarrow \infty$, which means that the first order derivative of the function will be positive at high SNR, and the proof is complete.
  \end{proof}

  By using the above lemma, we can     prove  that the use of NOMA improves the secrecy performance compared to OMA, as shown in the following theorem.
  \begin{theorem}\label{theorem2}
  The secrecy unicasting rate achieved by the NOMA scheme is always larger than or equal to  that of OMA, i.e., the following inequality always holds
  \begin{align}
R_S\geq \bar{R}_S,
  \end{align}
  in the high SNR regime.
  \end{theorem}
  \begin{proof}
 To prove the theorem, it is sufficient to prove the following inequality,   $\Delta_S\triangleq \left(R_{U,1} - \max\{R_{U,k}, 2\leq k\leq K\}\right)- \left(\bar{R}_{U,1} - \max\{\bar{R}_{U,k}, 2\leq k\leq K\}\right)\geq 0$. Without loss of generality, assume that the channels of the $(K-1)$ users (eavesdroppers for unicasting) are ordered as follows:
 \begin{align}
 z_2\geq \cdots \geq z_K.
 \end{align}
 Note that this assumption is used only to simplify the description of the proof. With this ordering, $u=z_K$.

 \subsubsection{When $\min\{z_1, z_K\}\leq \frac{\epsilon_M}{\rho}$}
 This case corresponds to the situation that all the power (time) will be allocated to multicasting, and no power (time) is available to unicasting, which means the unicasting rates are zero, $ R_{U,k}=  \bar{R}_{U,k}  =0$, for $1\leq k\leq K$, and the difference between the two secrecy rates is   zero.

 \subsubsection{When $\min\{z_1, z_K\}> \frac{\epsilon_M}{\rho}$ and $z_1<z_2$} In this case, at least one of the eavesdroppers has a better channel condition than  user $1$. Both   $\left(R_{U,1} - \max\{R_{U,k}, 2\leq k\leq K\}\right)$ and $\left(\bar{R}_{U,1} - \max\{\bar{R}_{U,k}, 2\leq k\leq K\}\right)$, are negative and therefore both the secrecy rates are zero, which means that the difference between two secrecy rates is still zero.

 \subsubsection{When $\min\{z_1, z_K\}> \frac{\epsilon_M}{\rho}$ and $z_1\geq z_2$} In this case, the difference between the two secrecy rates can be expressed as follows:
    \begin{align}\label{delta}
  \Delta_S = & \log\left(1+\rho z_1    \frac{z_{K}-\frac{\epsilon_{M}}{\rho}}{z_{K}(1+\epsilon_{M})}\right)  \\\nonumber & - \log\left(1+\rho z_2    \frac{z_{K}-\frac{\epsilon_{M}}{\rho}}{z_{K}(1+\epsilon_{M})}\right)-\left(1-   \frac{R_{M}}{\log(1+\rho z_K)} \right)
  \\\nonumber &\times \left(\log\left(1+\rho z_1 \right) -\log\left(1+\rho z_2 \right)\right),
  \end{align}
  where we use the assumption that the users have been ordered, i.e., $z_2$ is the largest channel gain and $z_K$ is the smallest among the $(K-1)$ users (eavesdroppers).  By using the function defined in \eqref{FU}, the secrecy rate difference can be expressed as follows:
      \begin{align}
  \Delta_S = & F_{z_K}(z_1) -  F_{z_K}(z_2).
  \end{align}
By applying Lemma \ref{lemma2}, we learn that $F_{z_K}(x)$ is a monotonically increasing function, which means $\Delta_S \geq 0$, since $z_1\geq z_2$.

In summary, the secrecy rate of $R_S$ is always larger than or equal to $\bar{R}_S$, and the proof is complete.
\end{proof}
{\it Remark 3:} The proof of Theorem \ref{theorem2} indicates that the event that  user $1$ has a weak channel gain results in the situation that OMA and NOMA have the same secrecy rates.  Following the rationales discussed in Section \ref{subsection spec 2}, we can again apply the proposed   user scheduling scheme to avoid this undesirable situation and improve the secrecy performance gain of NOMA over OMA.

{\it Remark 4:} Consider that the eavesdroppers' channels are ordered as $z_2\geq \cdots \geq z_K$, and assume that $\min\{z_1, z_K\}> \frac{\epsilon_M}{\rho}$ and $z_1\geq z_2$, i.e., the third case in the proof of Theorem \ref{theorem2}.   According to \eqref{delta} in the proof for Theorem \ref{theorem2}, $F_{z_K}(z_k)$ denotes the difference of user $k$'s capabilities to decode the unicasting message in the NOMA and OMA modes.
  It is worth pointing out that $F_{z_K}(z_K)=0$ since
\begin{align}
\log\left(1+\rho        \frac{z_K-\frac{\epsilon_{M}}{\rho}}{ (1+\epsilon_{M})}    \right) &= \log\left(        \frac{1+\rho z_K}{ (1+\epsilon_{M})}    \right) \\\nonumber
&= \log\left(1+\rho z_K \right)-  R_M.
\end{align}
Therefore  $F_{z_K}(z_k)\geq 0$, for $1\leq k \leq (K-1)$, since  $z_k\geq z_K$ and $F_{z_K}(z_k)\geq F_{z_K}(z_K)$. The fact that $F_{z_K}(z_k)\geq 0$ means that the use of NOMA can increase all the users' capabilities to detect the unicasting message. However,  the use of NOMA brings more improvements to user $1$ than other users, as pointed out  in Theorem \ref{theorem2}.

 \subsection{Characterizing the Secrecy Outage Probability}
Recall that the secrecy rate achieved by the NOMA scheme is given by
\begin{align}
R_S \triangleq &\left(\log\left(1+\rho z_1 \alpha_{U}^2\right)-\log\left(1+\rho v \alpha_{U}^2\right)\right)^+,
\end{align}
where $v=\max\{z_2, \cdots, z_K\}$.

Therefore the secrecy outage probability can be expressed as follows:
\begin{align}
\mathrm{P}_S \triangleq &\mathrm{P}\left(\log \left(1+\rho z_1 \alpha_{U}^2\right)-\log\left(1+\rho v \alpha_{U}^2\right)<\tilde{R}_S\right) ,
\end{align}
where $\tilde{R}_S$ is the targeted secrecy rate. This secrecy outage probability can be rewritten as follows:
\begin{align}
\mathrm{P}_S = &\mathrm{P}\left( (z_1-2^{\tilde{R}_S}v) \alpha_{U}^2 <\frac{\epsilon_S}{\rho} \right) ,
\end{align}
where $\epsilon_S=2^{\tilde{R}_S}-1$. By studying the relationship between $z_1$ and $z_k$, the outage probability can be further expressed as follows:
\begin{align}\label{qx3}
\mathrm{P}_S = &\underset{Q_{5}}{\underbrace{\mathrm{P}(z_1<u)+\mathrm{P}\left(z_1>u, u<\frac{\epsilon_M}{\rho}   \right)}}\\\nonumber &+\underset{Q_{6}}{\underbrace{\mathrm{P}\left(z_1>u, u>\frac{\epsilon_M}{\rho}  , z_1<2^{\tilde{R}_S}v \right)}}+Q_4,
\end{align}
where $u=\min\{z_2, \cdots, z_K\}$ and  the factor, $Q_4$, is expressed as follows:
\begin{align}\nonumber
Q_4 =&  \mathrm{P}\left(z_1>u>\frac{\epsilon_M}{\rho} , z_1>2^{\tilde{R}_S}v,
\right.\\\nonumber &\left.(z_1-2^{\tilde{R}_S}v)
   \frac{u-\frac{\epsilon_{M}}{\rho}}{u(1+\epsilon_{M})}    <\frac{\epsilon_S}{\rho} \right).
\end{align}
Note that for the three cases, $\{z_1<u\}$ and $\{z_1>u, u<\frac{\epsilon_M}{\rho} \}$, and $\{z_1>u>\frac{\epsilon_M}{\rho} , z_1<2^{\tilde{R}_S}v\}$  $ \mathrm{P}\left( (z_1-2^{\tilde{R}_S}z_k) \alpha_{U}^2 <\frac{\epsilon_S}{\rho} \right) =1$. The three terms, $Q_4$, $Q_5$ and $Q_6$, are calculated in the following subsections, respectively.

\subsubsection{Calculating $Q_4$}
Note that the largest and smallest channel gains, $v$ and $u$, are correlated as follows:\footnote{Without loss of generality, we focus on the cases with $K>2$. }
\begin{align}\label{uv}
f_{u,v}(u,v) &= (K-1)(K-2)e^{-u-v}\left(e^{-u}-e^{-v}\right)^{K-3}\\ \nonumber &=  \sum^{K-3}_{m=0}\tau_me^{-(K-2-m)u}e^{-(m+1)v},
\end{align}
where $\tau_m=(K-1)(K-2){K-3 \choose m} (-1)^m $.

We first rewrite the term, $Q_4$, as follows:
\begin{align}
Q_4 =&  \underset{v>u, u>\frac{\epsilon_M}{\rho}}{\mathcal{E}}\left\{\mathrm{P}\left(z_1>2^{\tilde{R}_S}v, \right.\right.\\ \nonumber &\hspace{4em}\left.\left.   (z_1-2^{\tilde{R}_S}v)
   \left(1-\frac{\epsilon_{M}}{u\rho}\right)    <\xi \right)\right\}\\
    =&\nonumber  \underset{v> u>\frac{\epsilon_M}{\rho}}{\mathcal{E}}\left\{\mathrm{P}\left(  2^{\tilde{R}_S}v< z_1<    2^{\tilde{R}_S}v+\frac{\xi}{\left(1-\frac{\epsilon_{M}}{u\rho}\right)}
     \right)\right\},
\end{align}
since
$2^{\tilde{R}_S}v>v>u$,
where $\xi=\frac{\epsilon_S(1+\epsilon_{M})}{\rho}$. %It is important to point out that $u<v$. The inequality is due to the fact that we omit the constraint of $v$,   $v>u$.

Therefore, $Q_4$ can be rewritten as follows:
\begin{align}
Q_4
    =&\nonumber \frac{1}{(M-1)!}  \underset{v> u>\frac{\epsilon_M}{\rho}}{\mathcal{E}}\left\{  \gamma\left(2^{\tilde{R}_S}v+\frac{\xi}{\left(1-\frac{\epsilon_{M}}{u\rho}\right)}\right) \right.\\   &\left.- \gamma \left(2^{\tilde{R}_S}v\right)
     \right\}.\label{diff1}
\end{align}
By using the joint pdf of $u$ and $v$ in \eqref{uv}, the term can be expressed as follows:
\begin{align}
Q_4
    &=  \sum^{K-3}_{m=0}\frac{\tau_m}{(M-1)!} \int^{\infty}_{\frac{\epsilon_M}{\rho}}e^{-(m+1)v}\int^{v}_{\frac{\epsilon_M}{\rho}} e^{-(K-2-m)u}\\ \nonumber
    &\times \left( \gamma\left(M,2^{\tilde{R}_S}v+\frac{\xi}{\left(1-\frac{\epsilon_{M}}{u\rho}\right)}\right) - \gamma \left(M,2^{\tilde{R}_S}v\right)
     \right)dudv.
\end{align}

Define the following function:
\begin{align}
&G(u,v)=e^{-(K-2-m)u}\\\nonumber
&\times\left( \gamma\left(M,2^{\tilde{R}_S}v+\frac{\xi}{\left(1-\frac{\epsilon_{M}}{u\rho}\right)}\right) - \gamma \left(M,2^{\tilde{R}_S}v\right)\right).
\end{align}
 The application of the Chebyshev-Gauss approximation yields the following expression:
\begin{align}
Q_4
    \approx&  \sum^{K-3}_{m=0}\frac{\tau_m}{(M-1)!} \sum^{N_a}_{i=1}w_i\sqrt{1-x_i^2}\int^{\infty}_{\frac{\epsilon_M}{\rho}}e^{-(m+1)v}\\ \nonumber &\times \frac{v-\frac{\epsilon_M}{\rho}}{2} G\left(\frac{v-\frac{\epsilon_M}{\rho}}{2} x_i+\frac{v+\frac{\epsilon_M}{\rho}}{2},v\right)dv.
\end{align}
The remaining  integration can be further approximated by using the Chebyshev-Gauss approximation as follows:
\begin{align}
Q_4\label{q44}
    \approx &  \sum^{K-3}_{m=0}\frac{\tau_m}{(M-1)!} \sum^{N_a}_{i=1}w_i\sqrt{1-x_i^2}\int^{\frac{\rho}{\epsilon_M}}_0e^{-\frac{(m+1)}{y}}\\ \nonumber &\times \frac{\frac{1}{y}-\frac{\epsilon_M}{\rho}}{2} G\left(\frac{\frac{1}{y}-\frac{\epsilon_M}{\rho}}{2} x_i+\frac{\frac{1}{y}+\frac{\epsilon_M}{\rho}}{2},\frac{1}{y}\right)y^{-2}dy\\ \nonumber
    =&  \sum^{K-3}_{m=0}\frac{\tau_m}{(M-1)!} \sum^{N_a}_{i=1}w_i\sqrt{1-x_i^2}\sum^{N_a}_{j=1}\frac{w_j\rho}{2\epsilon_M} e^{-\frac{(m+1)}{\tilde{y}}}\\ \nonumber &\times \frac{\frac{1}{\tilde{y}}-\frac{\epsilon_M}{\rho}}{2} \left(\tilde{y}
    \right)^{-2}\sqrt{1-y_j^2}\\ \nonumber &\times G\left(\frac{\frac{1}{\tilde{y}}-\frac{\epsilon_M}{\rho}}{2} x_i+\frac{\frac{1}{\tilde{y}}+\frac{\epsilon_M}{\rho}}{2}, \frac{1}{\tilde{y}}\right),
\end{align}
where $w_j=\frac{\pi}{N_a}$, $\tilde{y}=\frac{\rho}{2\epsilon_M}y_j+\frac{\rho}{2\epsilon_M}$ and $y_j =\cos\left(\frac{2j-1}{2N_a}\pi\right)$

\subsubsection{Calculating $Q_6$}On the other hand, the third term in \eqref{qx3} can be found as follows:
\begin{align}\label{diff2}
Q_6 =&\underset{v>u>\frac{\epsilon_M}{\rho}}{\mathcal{E}}\left\{\mathrm{P}\left(u<z_1  <2^{\tilde{R}_S}v \right)\right\}
\\ \nonumber =&\frac{1}{(M-1)!}\underset{v>u>\frac{\epsilon_M}{\rho}}{\mathcal{E}}\left\{\gamma(M, 2^{\tilde{R}_S}v) - \gamma(M,u)\right\}.
\end{align}
Although  each component inside of the expectation is only a function of either $u$ or $v$, it is important to point out that this expectation cannot  be simply  evaluated  as follows:
\begin{align}
Q_6 \neq&\frac{1}{(M-1)!}\underset{v>\frac{\epsilon_M}{\rho}}{\mathcal{E}}\left\{\gamma(M, 2^{\tilde{R}_S}v)\right\} -\underset{u>\frac{\epsilon_M}{\rho}}{\mathcal{E}}\left\{ \gamma(M,u)\right\},
\end{align}
which is due to the implicit constraints that both $u$ and $v$ are larger than $\frac{\epsilon_M}{\rho}$.

Following steps similar to those for calculating   $Q_4$, we first define
\begin{align}
G_2(u,v)=& e^{-(K-2-m)u}\left( \gamma\left(M,2^{\tilde{R}_S}v\right) - \gamma \left(M,u\right)\right).
\end{align}
$Q_6$ can then be written as follows:
\begin{align}\label{q6}
Q_6
    \approx &   \sum^{K-3}_{m=0}\frac{\tau_m}{(M-1)!} \sum^{N_a}_{i=1}w_i\sqrt{1-x_i^2}\sum^{N_a}_{j}\frac{w_j\rho}{2\epsilon_M} e^{-\frac{(m+1)}{\tilde{y}}}\\ \nonumber &\times \frac{\frac{1}{\tilde{y}}-\frac{\epsilon_M}{\rho}}{2} \left(\tilde{y}
    \right)^{-2}\sqrt{1-y_j^2}\\ \nonumber &\times G_2\left(\frac{\frac{1}{\tilde{y}}-\frac{\epsilon_M}{\rho}}{2} x_i+\frac{\frac{1}{\tilde{y}}+\frac{\epsilon_M}{\rho}}{2}, \frac{1}{\tilde{y}}\right).
\end{align}
It is worth pointing out that the expression of $Q_6$ is quite similar to that of $Q_4$, which is due to the similarity between \eqref{diff1} and \eqref{diff2}.

\subsubsection{Calculating $Q_5$}
$Q_5$ is a sum of two probabilities as shown in the following:
\begin{align}
Q_5=&\mathrm{P}(z_1<u)+\mathrm{P}\left(z_1>u, u<\frac{\epsilon_M}{\rho}   \right)\\\nonumber
=&\mathrm{P}(z_1<u)+\mathrm{P}\left(z_1>u   \right)-\mathrm{P}\left(z_1>u, u>\frac{\epsilon_M}{\rho}   \right)
\\\nonumber
=&1-\mathrm{P}\left(z_1>u, u>\frac{\epsilon_M}{\rho}   \right).
\end{align}
Therefore,
\begin{align}\label{q5}
Q_5 =&1-    \frac{\Gamma(M, \frac{\epsilon_M}{\rho})}{(M-1)!}e^{-\frac{\epsilon_M(K-1)}{\rho}} +\frac{K^{-M}\Gamma(M, \frac{\epsilon_M K}{\rho})}{(M-1)!}.
\end{align}
By substituting \eqref{q44}, \eqref{q6} and \eqref{q5} into  \eqref{qx3}, an approximated expression for the secrecy outage probability is obtained.

{\it Remark 4:}
As can be seen from \eqref{qx3}, the secrecy outage probability consists of three parts. The term $Q_6$ is dominant, compared to $Q_4$ and $Q_5$, particularly at high SNR and when $M$ and $K$ are large. Specifically,  $Q_5$ can be approximated as follows:
\begin{align}
Q_5 \approx&  K^{-M},
\end{align}
 which is quite small when $K$ and $M$ are large. As shown in \eqref{diff1},
 $Q_4$ is related to the difference between  the following gamma functions: $ \gamma\left(2^{\tilde{R}_S}v+\frac{\xi}{\left(1-\frac{\epsilon_{M}}{u\rho}\right)}\right)$ and $ \gamma \left(2^{\tilde{R}_S}v\right)$. For fixed $u$ and $v$, increasing SNR can reduce the difference between the two functions, and hence reduce the value of $Q_4$. On the other hand, $Q_6$ is related to the difference between the two following functions: $\gamma(M, 2^{\tilde{R}_S}v)$ and $\gamma(M,u)$. When $K$ is large, this difference can be very large, which means a large value for $Q_6$. Note that this difference cannot be reduced by simply increasing the SNR.

{\it Remark 5:} The use of the user scheduling scheme described in Section \ref{subsection spec 2} is still helpful to reduce the secrecy outage probability, as explained in the following. Recall that the term $Q_6$ is dominant in the expression of the outage probability. By using the proposed user scheduling scheme, the channel gain of the user selected for unicasting will be very strong, which makes the event $z_1< 2^{\tilde{R}_Sv}$ less likely. As shown in \eqref{qx3},  this will  reduce the value of $Q_6$, and hence improve  the overall outage probability.

\section{Numerical Results}\label{section simulation}
In this section, the spectral efficiency and security performance of the proposed NOMA transmission scheme is demonstrated by using simulation results.
\begin{figure}[!htp]\vspace{-0.5em}
\begin{center} \subfigure[Unicasting Outage Rates -- $(1-\mathrm{P}^o)R_U$]{\label{cooperative NOMA_system1}
\includegraphics[width=0.40\textwidth]{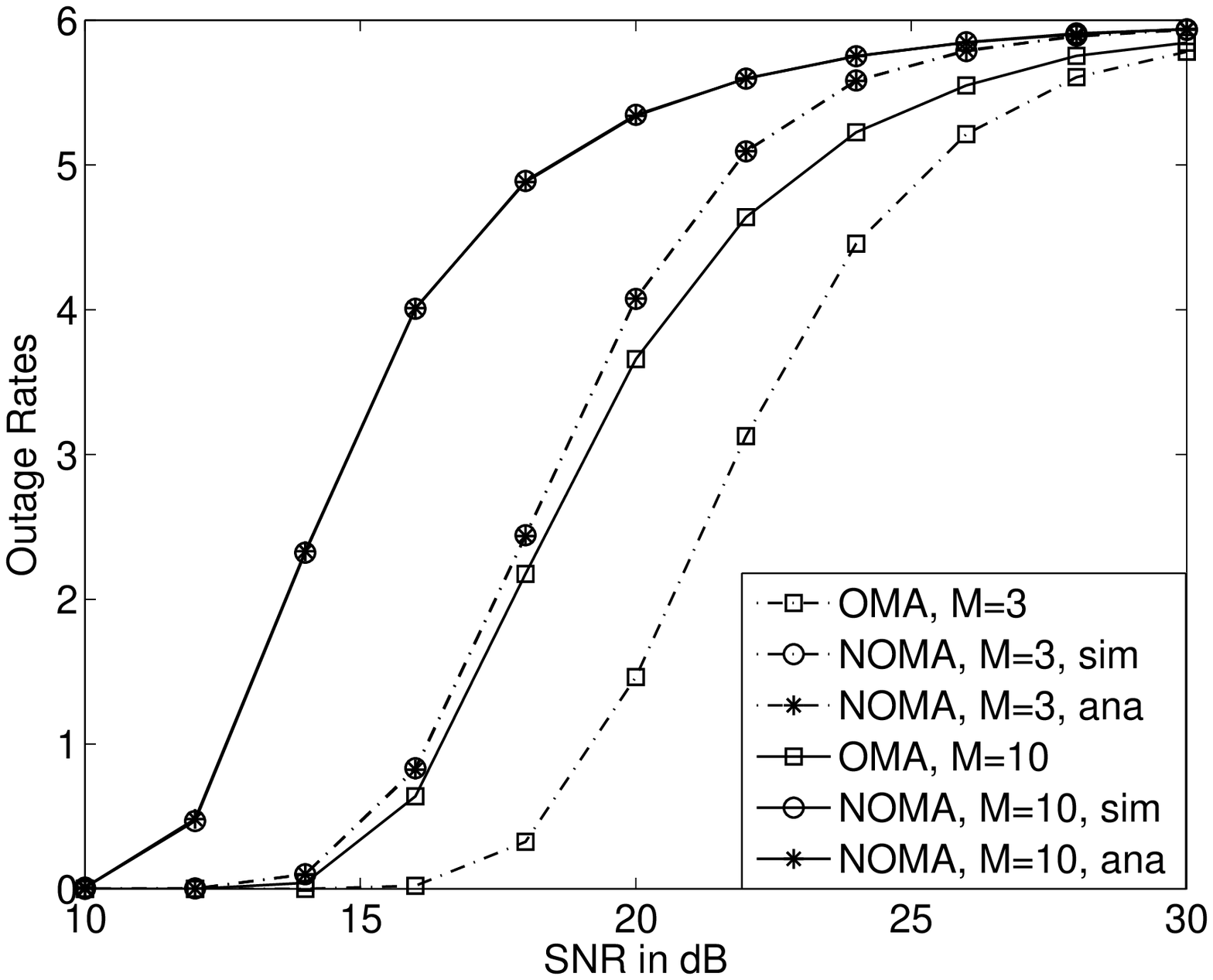}}
\subfigure[Outage probability for unicasting]{\label{cooperative NOMA_poor user2}\includegraphics[width=0.40\textwidth]{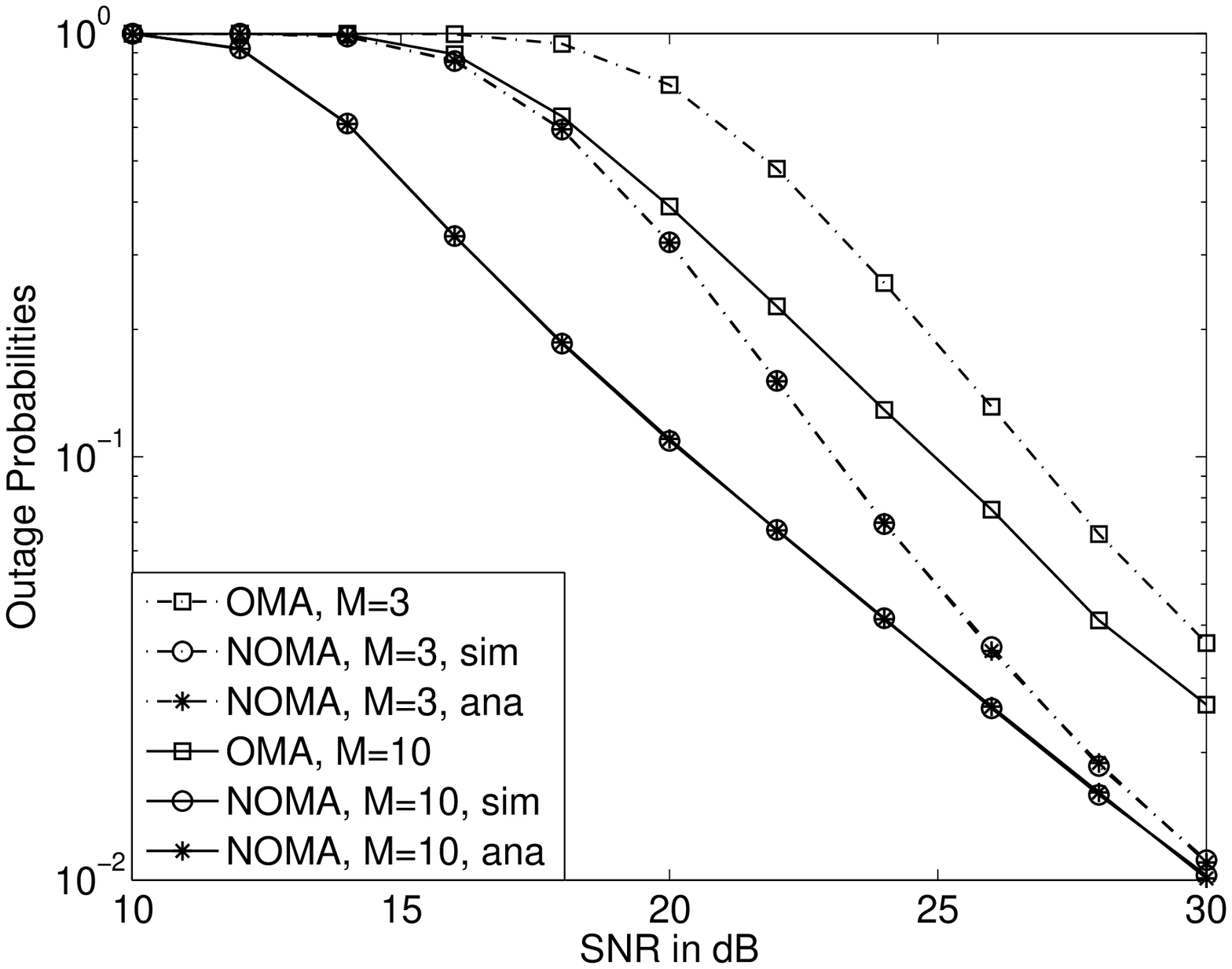}}\vspace{-1em}
\end{center}
  \caption{Performance comparison between the OMA and NOMA transmission schemes.  $K=11$ and $N_a=20$. The targeted data rates for multicasting and unicasting are $1$ and $6$ bits per channel use (BPCU), respectively.  }\vspace{-0.5em}
   \label{fig1}
\end{figure}

In Fig. \ref{fig1}, the unicasting outage probability and outage rate achieved by NOMA are compared with those of OMA. As can be seen from the figures, by using more antennas at the base station, both the outage rates and probabilities for  NOMA and OMA are improved. In addition, the figures show that the use of NOMA can significantly improve the unicasting rates, compared to OMA. For example, when the SNR is $16$dB and $M=10$, the use of NOMA can support a unicasting rate of $4$ bits per channel use (BPCU), whereas OMA can support a rate of $0.8$ BPCU only, i.e., the NOMA unicasting rate is nearly $5$ times the OMA rate. Similar performance gains in terms of the outage probability can also be observed from Fig. \ref{cooperative NOMA_poor user2}.

It is important to point out that such a significant gain is obtained without degrading the multicasting performance. Particularly,  NOMA realizes the same multicasting performance as OMA, as shown in Proposition \ref{propostion 1}.  In addition, both figures also demonstrate the accuracy of the developed analytical results, whereas the curves for the simulation results match perfectly those for the analytical results. In addition, Lemma \ref{lemma1} shows that a diversity gain of $1$ is achieved, no matter how many antennas the base station has. This is also confirmed by Fig. \ref{cooperative NOMA_poor user2}, since the slope of all the curves becomes the same at high SNR.

\begin{figure}[thb]
\begin{center}
\includegraphics[width=0.40\textwidth]{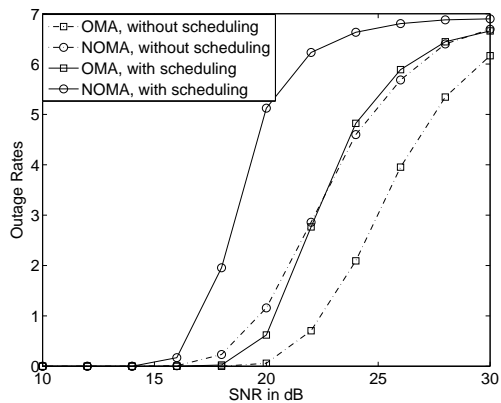}
\end{center}
\caption{The impact of scheduling on the performance of unicasting.  $M=2$ and $K=11$. The targeted data rates for multicasting and unicasting are $1$ and $7$ BPCU, respectively. }
        \label{Fig2}
\end{figure}

In Section \ref{subsection spec 2}, a user scheduling scheme was proposed in order to further improve the performance gap between the NOMA and OMA transmission schemes. This performance enhancement due to the use of user scheduling can be clearly observed in Fig. \ref{Fig2}. Particularly, the use of user scheduling can bring performance improvements to both NOMA and OMA, but  NOMA benefits more from user scheduling  than OMA.  For example, the outage rate curve for NOMA is shifted to the left nearly $4$dB, whereas the one for OMA is shifted to the left around $2$dB. In this paper, we have used the same beamforming for both the NOMA and OMA modes. As discussed in Section \ref{subsection beamforming}, one can also use random  beamforming or equal gain combining based beamforming, instead of the choice shown in \eqref{bemforming}. Fig. \ref{Fig3} demonstrates that the difference between the OMA schemes with different beamforming is insignificant, and the use of the choice in \eqref{bemforming} offers a slight performance gain over the others.

\begin{figure}[thb]
\begin{center}
\includegraphics[width=0.40\textwidth]{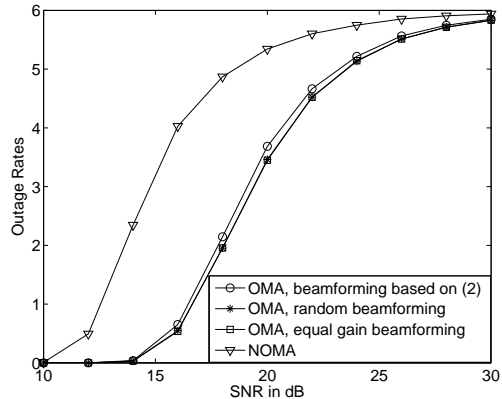}
\end{center}
\caption{The impact of scheduling on the performance of unicasting.  $M=10$ and $K=11$. The targeted data rates for multicasting and unicasting are $1$ and $6$ BPCU, respectively.}
        \label{Fig3}
\end{figure}

\begin{figure}[!htp]\vspace{-1em}
\begin{center} \subfigure[Secrecy outage Rates -- $(1-\mathrm{P}^o)\tilde{R}_S$]{\label{cooperative NOMA_system3}
\includegraphics[width=0.40\textwidth]{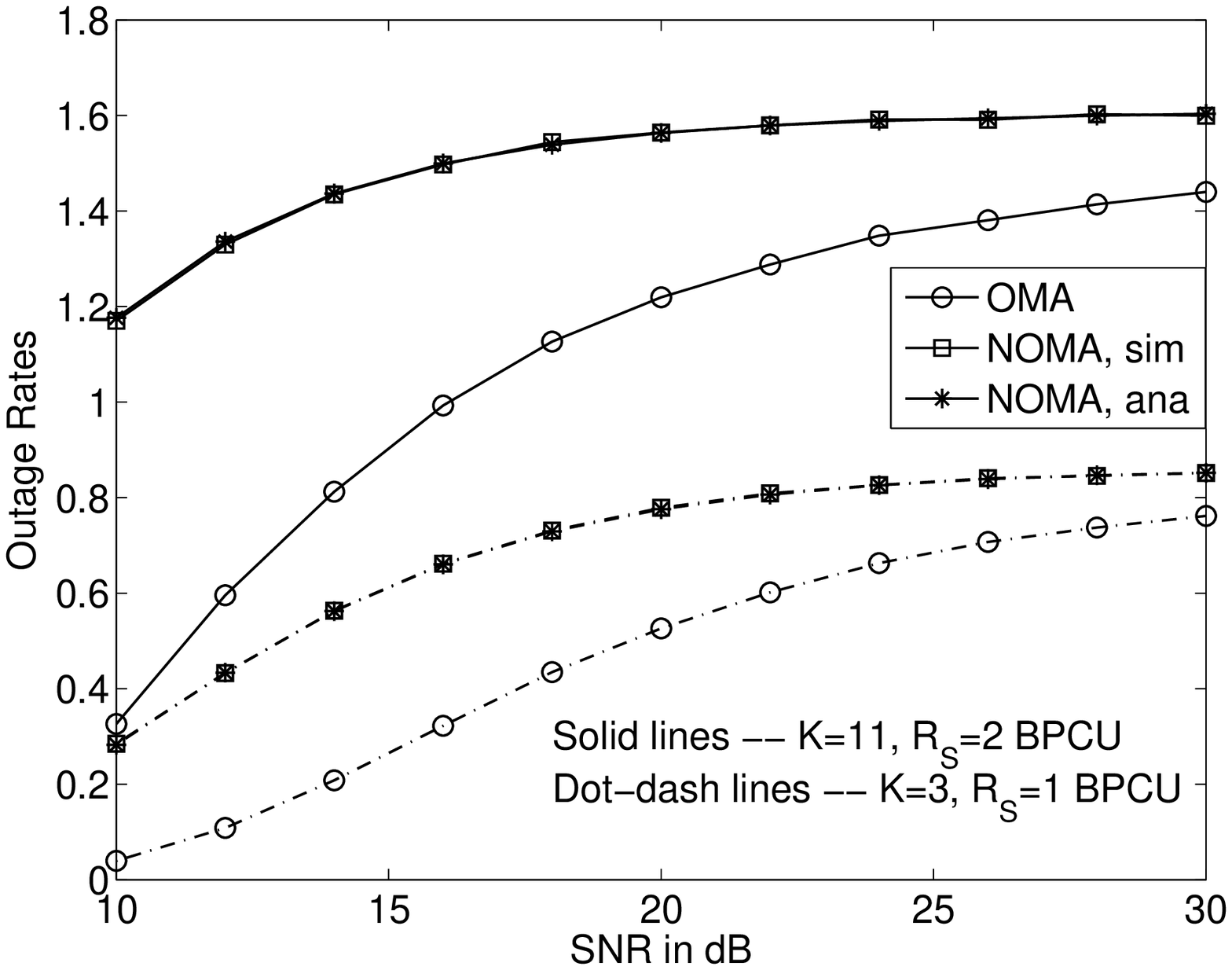}}
\subfigure[Outage probability for secrecy unicasting]{\label{cooperative NOMA_poor user4}\includegraphics[width=0.40\textwidth]{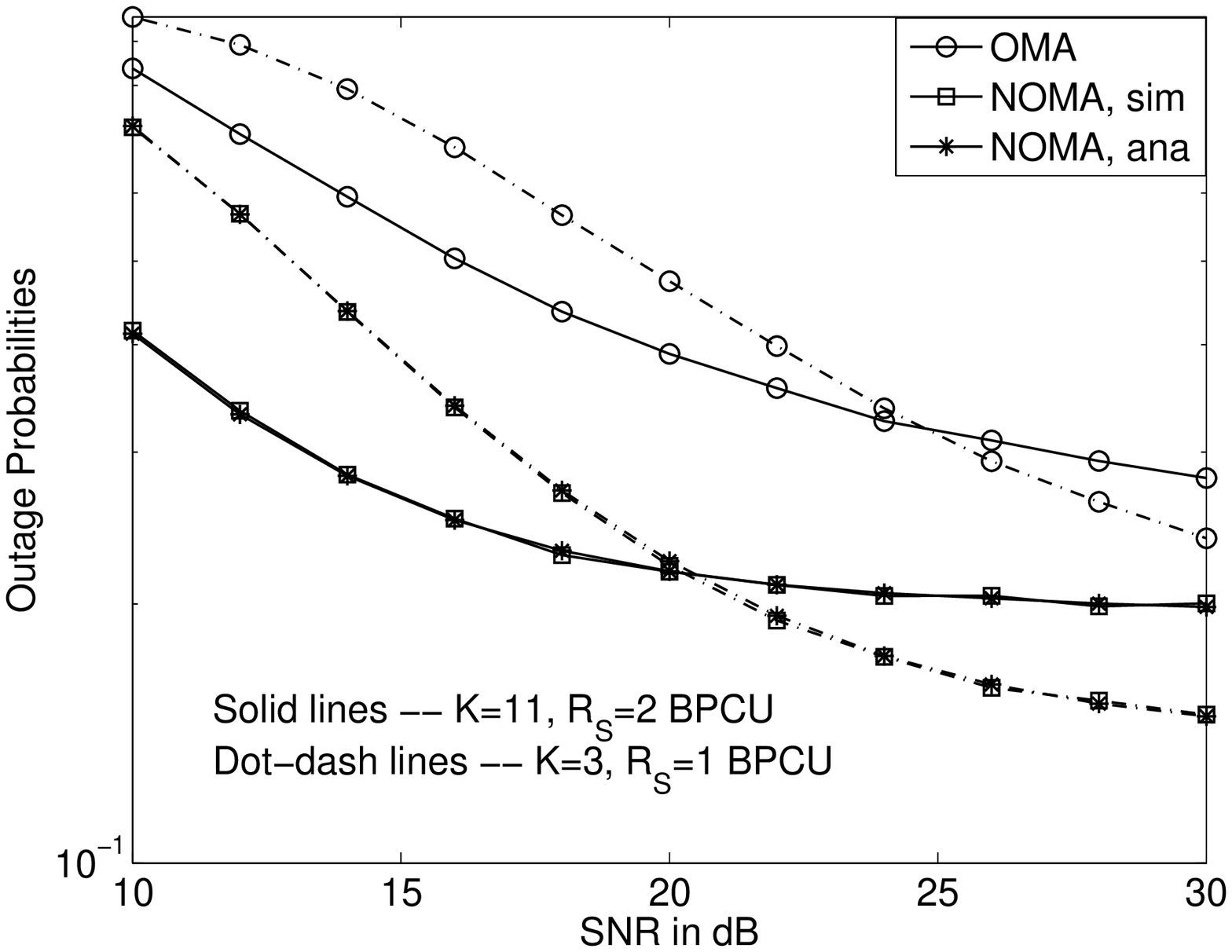}}\vspace{-1em}
\end{center}
  \caption{Secrecy performance comparison between the OMA and NOMA transmission schemes.  $M=10$ and $N_a=500$. The targeted data rate for multicasting is $1$ BPCU.  }\vspace{-1em}
   \label{fig4}
\end{figure}

\begin{figure}[thb]
\begin{center}
\includegraphics[width=0.40\textwidth]{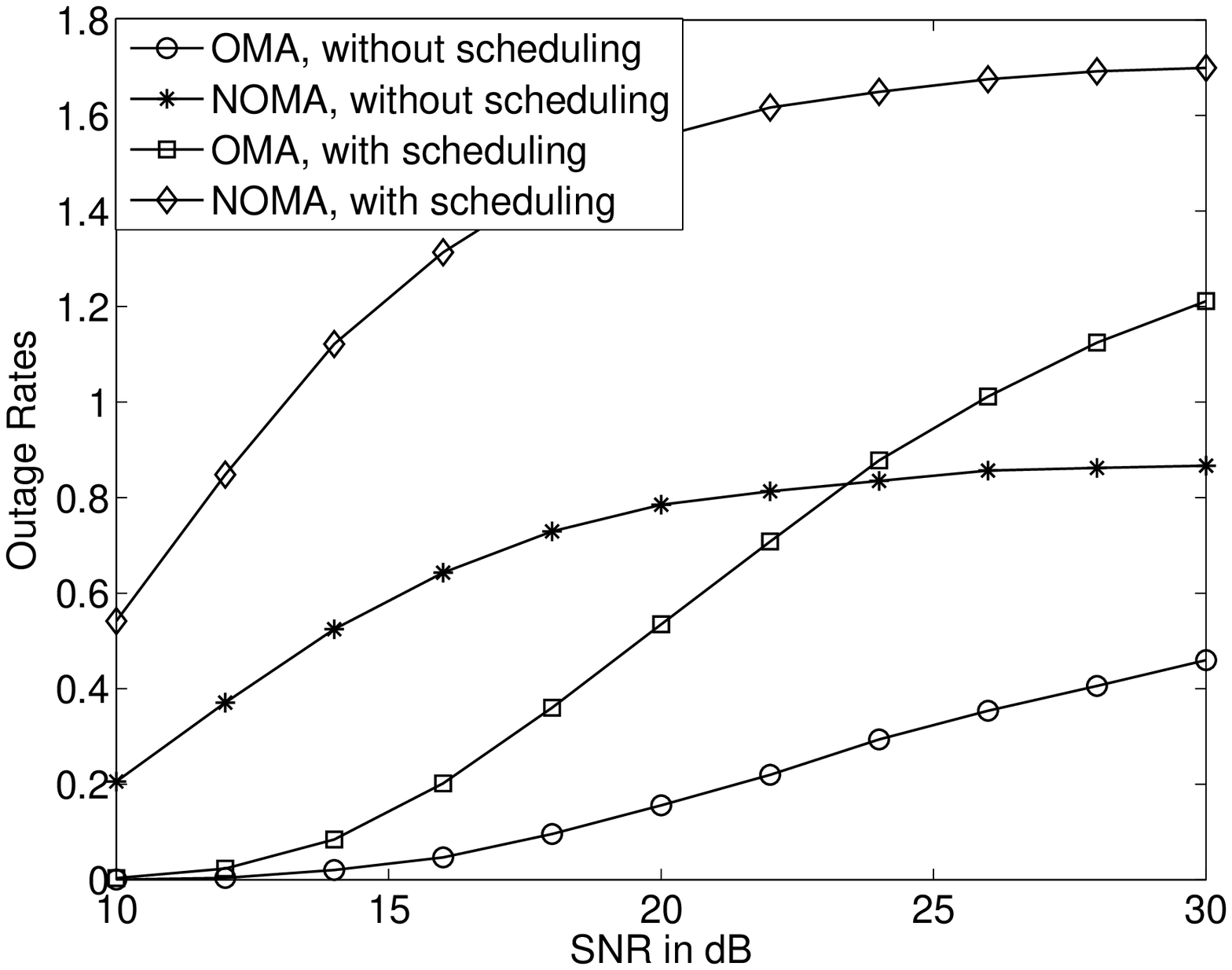}
\end{center}
\caption{The impact of scheduling on the performance of secrecy unicasting.  $M=10$ and $K=11$. The targeted data rates for multicasting and secrecy unicasting are $1$ and $2$ BPCU, respectively. }
        \label{Fig5}
\end{figure}

 In Fig. \ref{fig4}, the secrecy performance of NOMA unicasting is demonstrated by using OMA as a benchmarking scheme. The simulation results shown in Fig. \ref{fig4} are consistent to the analytical results developed in Section \ref{section security}. For example, Theorem \ref{theorem2} shows that the secrecy unicasting  rate of NOMA is always larger than or equal to that of OMA, which is confirmed by Fig. \ref{fig4}. Particularly, the rate performance gain of NOMA is   clearly demonstrated, e.g., a secrecy rate of $1.2$ BPCU can be achieved by NOMA at a SNR of $10$dB, and this is significantly larger than $0.3$ BPCU, a rate achieved by OMA. In addition, the curves for the simulation results match those for the analytical results, which verifies the accuracy of the developed analytical results. In Fig. \ref{Fig5}, the impact of user scheduling on the performance of secrecy unicasting is demonstrated. As can be observed from the figure, the use of user scheduling can significantly enlarge the performance gap between NOMA and OMA. For example, when the SNR is $20$dB, the performance gap between NOMA and OMA is $0.6$ BPCU without using the user scheduling scheme, and this gap can be increased to $1.1$ BPCU when the user scheduling scheme is applied.

\section{Conclusions}
 In this paper, the application of NOMA to a multi-user network with mixed multicast and unicast traffic has been considered. Beamforming and power allocation coefficients have been jointly designed to ensure that the unicasting performance is improved while maintaining the reception reliability of multicasting. Both analytical and simulation results have been developed  to demonstrate that the use of the NOMA assisted multicast-unicast scheme yields a significant improvement in spectral efficiency compared to orthogonal multiple access (OMA) cases, in which multicasting and unicasting are realized separately. Since the unicasting message is broadcasted to all the users, how well the use of NOMA can prevent those multicasting receivers intercepting the unicasting message has  also been investigated, where it is shown that the secrecy unicasting rate achieved by NOMA is always larger than or equal to that of OMA.

 \bibliographystyle{IEEEtran}
\bibliography{IEEEfull,trasfer}
  \end{document}